\documentclass{elsarticle}

\def\appendixname{Appendix }
\makeatletter
\renewcommand\appendix{\par
  \setcounter{section}{0}%
  \setcounter{subsection}{0}%
  \setcounter{equation}{0}
  \gdef\thefigure{\@Alph\c@section.\arabic{figure}}%
  \gdef\thetable{\@Alph\c@section.\arabic{table}}%
  \gdef\thesection{\@Alph\c@section}%
  \@addtoreset{equation}{section}%
  \gdef\theequation{\@Alph\c@section.\arabic{equation}}%
  \addtocontents{toc}{\string\let\string\numberline\string\tmptocnumberline}{}{}
}
\makeatother

\usepackage{hyperref}
\usepackage{amsthm,amsmath,amssymb}
\usepackage{tikz}
\usepackage{pgfplots}
\pgfplotsset{compat=1.13} \usepackage{lastpage}
\usepackage{acronym}
\usepackage[capitalize]{cleveref}
\usepackage[mode=buildnew]{standalone}\usepackage{setspace}
\usepackage{xpatch}

\newacro{TpI}{tests per individual}

\newtheoremstyle{theorem}{}{}{\itshape}{}{\bf}{.}{ }{}%
\newtheoremstyle{example}{}{}{} {}{\bf}{.}{ }{}%

\theoremstyle{theorem}
\newtheorem{theorem}{Theorem}[section]

\theoremstyle{example}
\newtheorem{example}[theorem]{Example}
\newtheorem{remark}[theorem]{Remark}

\theoremstyle{theorem}
\newtheorem{lemma}[theorem]{Lemma}
\newtheorem{definition}[theorem]{Definition}
\newtheorem{corollary}[theorem]{Corollary}

\crefname{equation}{Equation}{Equations}
\Crefname{equation}{Equation}{Equations}
\crefformat{equation}{#2(#1)#3}
\crefrangeformat{equation}{#3(#1)#4--#5(#2)#6}
\crefmultiformat{equation}{#2(#1)#3}{ and~#2(#1)#3}{, #2(#1)#3}{ and~#2(#1)#3}
\Crefformat{equation}{#2(#1)#3}
\Crefrangeformat{equation}{#3(#1)#4--#5(#2)#6}
\Crefmultiformat{equation}{#2(#1)#3}{ and~#2(#1)#3}{, #2(#1)#3}{ and~#2(#1)#3}
\crefname{figure}{Figure}{Figures}
\Crefname{figure}{Figure}{Figures}
\crefname{appendix}{Appendix}{Appendices}
\Crefname{appendix}{Appendix}{Appendices}

\AtBeginEnvironment{appendix}{
  \crefalias{section}{appendix}
  \crefalias{subsection}{appendix}
}

\def\dfink{D^{(\mathbf{X}, \rho)}(K)}

\DeclareMathOperator*{\argmin}{arg\,min}

\pgfkeys{/pgf/number format/.cd,1000 sep={\,}}

\pgfplotsset{
  DRaxis/.style={
    width=.95\textwidth,
    height=11cm,
    xlabel={Tests per $100\,000$ individuals},
    ylabel={Cost per individual},
    ylabel style={
                },
    tick label style={font=\small},
    grid=major, 
    legend cell align=left,
    legend style={
      legend pos=north east, font=\small
    },
  },
  DRaxisTpI/.style={
    xlabel={Tests in \ac{TpI}},
  },
  DRplot/.style={
    no markers,
    line width=1pt,
  },
  DRlabeled/.style={
    mark=x,
    mark options={fill=blue,scale=2},
    line width=1pt,
    nodes near coords,     visualization depends on={value \thisrow{anchor}\as\anchor},
    every node near coord/.append style={anchor=\anchor,font=\footnotesize},
        point meta=explicit symbolic,   },
  table/DRtableTpI/.style={
    x=R,
    y=D,
  },
  table/DRtable/.style={
    DRtableTpI,
    x expr=\thisrow{R}*100000,
  },
  table/DRtableTpI_label/.style={
    DRtableTpI,
    meta=label,
  },
  table/DRtable_label/.style={
    DRtable,
    meta=label,
  },
}

\newcommand{\ourtitle}{
Modelling the Utility of Group Testing for Public Health Surveillance
}

\title{\ourtitle}

\author[1]{G\"unther~Koliander\corref{cor1}}
\ead{gkoliander@kfs.oeaw.ac.at}
\author[2]{Georg Pichler}
\ead{georg.pichler@ieee.org}

\cortext[cor1]{Corresponding author}

\allowdisplaybreaks

\begin{document}

\affiliation[1]{organization={Acoustics Research Institute, 
Austrian Academy of Sciences},
addressline={Wohllebengasse~12-14},
postcode={1040},
city={Vienna},
country={Austria}}

\affiliation[2]{organization={Institute of Telecommunications, TU Wien},
addressline={Gußhausstraße~25/E389},
postcode={1040},
city={Vienna},
country={Austria}}

\begin{abstract}
  In epidemic or pandemic situations, resources for testing the infection status of individuals may be scarce. 
Although group testing can help to significantly increase testing capabilities, the (repeated) testing of entire populations can exceed the resources of any country.
We thus propose an extension of the theory of group testing that takes into account the fact that definitely specifying the infection status of each individual is impossible. 
Our theory builds on assigning to each individual an infection status (healthy/infected), as well as an associated cost function for erroneous assignments. 
This cost function is versatile, e.g., it could take into account that false negative assignments are worse than false positive assignments and that false assignments in critical areas, such as health care workers, are more severe than in the general population.
Based on this model, we study the optimal use of a limited number of tests to minimize the expected cost.
More specifically, we utilize information-theoretic methods to give a lower bound on the expected cost and describe simple strategies that can significantly reduce the expected cost over currently known strategies.
A detailed example is provided to illustrate our theory.
\end{abstract}
\begin{keyword}
group testing \sep public health surveillance \sep source coding \sep rate-distortion theory
\end{keyword}

\maketitle

\section{Introduction}

The current pandemic revitalized research on group testing, a methodology to reduce the number of required tests to screen a large population for a certain disease.
The increased testing efficiency results from jointly testing groups, instead of subjecting each individual to a test.

Specifically, we consider the scenario of \textit{probabilistic} group testing as first described by \citet{dorfman43}, where every individual has a certain (known) probability to be infected. 
Probabilistic group testing
is in contrast to combinatorial group testing \cite{Du2000Combinatorial,lee2019saffron} where one assumes a fixed known number of infected individuals. 
Furthermore, we assume that while the number of tests is limited, all tests are perfectly accurate.
Although the case of imprecise tests also leads to interesting resource allocation problems \cite{ely2020optimal}, we will not consider it in this work.
Thus, if none of the individuals in a tested group are infected, the test will yield a negative result with certainty and one test was sufficient to obtain a definite result for several individuals. 
If, however, the test is positive, one cannot say which individuals in the group are infected and further tests have to be conducted.

Formally, a \textit{testing strategy} is a deterministic algorithm, describing which groups of individuals are pooled and jointly subjected to the test.
We assume that the choice of the next group can depend on the results of previously conducted tests.
This \textit{adaptive} group testing, is in contrast to non-adaptive group testing (for a recent survey see \citet{aldridge2019group}), where the testing strategy is fixed and all tests can thus be performed in parallel.

\subsection*{Previous work}
The theoretical work on adaptive, probabilistic group testing started with \citet{dorfman43} who considered the simple strategy of testing groups of $M$ individuals and subsequently, if the pooled test is positive, testing each individual in the group separately.
Although this strategy is far from optimal, it is easy to implement and can significantly increase the testing capabilities, compared to testing every sample individually.
Further research in this direction was focused on finding new, more efficient and practical testing strategies, e.g., nested testing \cite{sobel1959group}, binary splitting \cite{hwang1972method}, and array testing \cite{phatarfod1994use}.
The binary splitting algorithms by \citet{hwang1972method} perform close to the theoretic optimum, within a factor of $1.11$ of an information-theoretic lower bound, as shown by \citet{aldridge2019rates}.
However, such a binary splitting technique requires extensive bookkeeping and many tests need to be performed sequentially, and cannot be conducted in parallel. 
Simpler, two-stage procedures are investigated by \citet{berger2002asymptotic}.
The close relation of adaptive, probabilistic group testing to variable length source coding is well-detailed by~\citet{wolf1985born}.

Due to the current pandemic, several works rediscover slight variations of these results and argue for the use of group testing \cite{gollier2020group,eberhardt2020multi}.
Also, the practical implementation of group testing has seen a new surge of research, in particular, questioning the practical use in testing for SARS-CoV-2.
Here, however, the focus was on the classical Dorfman testing strategy \cite{abdalhamid2020assessment,hogan2020sample,yelinevaluation} and only few works considered more elaborate non-adaptive testing strategies \cite{shental2020efficient}.

Hardly any  works go beyond the assumptions of perfect sensitivity and independence of the infection status of tested individuals.
However, \citet{pilcher2020group} takes the dilution effect into account, i.e., reduced sensitivity of tests for large groups,  
and \citet{deckert2020simulation} performed a simulation study that found benefits of group testing if there is positive correlation of infection status between individuals in the same group.

The common ground of all works above is that they focus on identifying exactly, which individuals are infected.
This is a valid strategy and clearly the best outcome. 
However, there can be situations where not sufficiently many tests are available to subject all potentially infected individuals to a test. 
This is particularly the case when many individuals are (potentially) infected and testing resources are scarce.
Here, a strategy needs to be found that uses these limited resources effectively.
When insufficient tests for the entire population are available, it is unavoidable that some healthy individuals will be treated as infected and/or some infected individuals treated as healthy. For the design of a testing strategy, it is important to note, that these two events are not equally harmful in general. It might be considerably worse to have an undetected infection (false negative) present, than to treat one healthy individual as infected (false positive).
A suitable balance has to be found and sensitive questions like ``How many false positives are we willing to accept to prevent one false negative?'' need to be answered to do so.
Here, assuming that quantitative answers to these questions can be given, we propose a framework for designing and evaluating group testing strategies.
Mathematically, the resulting problem is one of rate-distortion theory, a branch of information theory, established in the 1950s in a seminal paper by~\citet{Shannon1959Coding}.
This connection allows us to formulate bounds on the performance of any group testing strategy.
Although the problem of (combinatorial) group testing was extensively explored by information-theorists (e.g., in Ch.~24--29 in \citet{Aydinian2013Information} and by \citet{aldridge2019group}), this rate-distortion viewpoint has apparently not been considered so far. 
Thus, fundamental bounds are missing even for elementary scenarios.

Only recently, researchers suggested that ``the  scarcity  of tests obviously means that it is better to use a test to detect the virus in another untested group than to try to discover who is infected in a positive group'' \cite{gollier2020group} which is a first step into the direction of the rigorous theory developed here.
Finally, the scenario discussed in~\cite{jonnerby2020maximising} is closest to our ideas but considers only a single fixed testing strategy. 
We will use it in our theory for comparison and as a starting point for some more evolved testing strategies.

\subsection*{Contributions}
We present a rigorous mathematical framework for evaluating the cost incurred by false positive and false negative assignments under a given group testing strategy.
An ultimate lower bound on the expected cost that cannot be surpassed by any testing strategy is derived and compared to existing and novel testing strategies. 
We consider two basic scenarios in more detail:
First, a toy example where all individuals are equally likely to be infected and where wrong assignments incur the same cost for each individual;
and, subsequently, a division of the population into subpopulations that have different probabilities of being infected (e.g., individuals showing symptoms are more likely to be infected than individuals without symptoms) and/or different costs associated to false assignments (e.g., misclassified health care workers result in a higher cost).

Our work is  focused on a simple model that requires as few parameters as possible.
Thus, it  does not capture several aspects that might be relevant in practical scenarios, such as dependence of the infection status between individuals, compliance of individuals with their assigned health status, or imperfect test results, nor does it incorporate testing strategies into a larger disease model.
Nevertheless, basic questions that were so far answered by ``common sense,'' can be discussed on a sound mathematical basis. 
For example, optimal testing priorities can be shown to depend heavily on the specific scenario and subjecting only symptomatic individuals to a test is often a suboptimal decision. 

The rest of this article is organized as follows. 
In \cref{sec_problem}, we formulate the problem, give a mathematical definition of testing strategies, and introduce the expected cost associated with a testing strategy, as well as the minimal expected cost.
Our main theoretical results are presented in \cref{sec:results}. 
We establish fundamental lower bounds on the minimal expected cost and calculate the expected cost of various simple testing strategies. 
A first simple example is given to illustrate the bound and the strategies.
\Cref{sec:case} showcases a more complicated example.
It illustrates how to allocate limited testing resources to obtain significant improvements using simple testing strategies. 
In \cref{sec:discussion}, we discuss our results and their limitations. 
Finally, in \cref{sec:itdetails}, we provide the information-theoretic statements that underlie our main results. Detailed proofs are relegated to a technical appendix.

\section{Problem formulation}\label{sec_problem}

We assume that we have a sequence of individuals and the infection status of the $n$-th individual is given by a binary random variable $X_n$ on $\Omega := \{0,1\}$, where $X_n=1$ corresponds to being infected and $X_n=0$ to being healthy. 
All $X_n$ are assumed to be independent but not necessarily identically distributed. 
Thus, each  $X_n$ is a Bernoulli($p_n$) random variable with possibly different probability $\Pr[X_n=1] = p_n \in (0, 1)$. 
The second ingredient we need is a \textit{cost function} $\rho_n(x_n, y_n)$ that models the cost of wrong assignments, i.e., assigning an estimated infection status $y_n$ to the $n$-th individual with actual infection status $x_n$.
In contrast to communication scenarios where $0$ and $1$ are usually interchangeable, the cost $\rho_n(0,1)$ of false positive assignment (i.e., a healthy individual is wrongly assigned an infected status) 
is not necessarily the same as the cost $\rho_n(1,0)$ of a false negative assignment (i.e., an infected individual is wrongly assigned a healthy status). 
Thus, we define
\begin{align}
    \rho_n(0,1) = b_n\,, && \text{ and} &&
    \rho_n(1,0) = c_n\,,
    \label{eq:false}
\end{align}
where $b_n, c_n > 0$.
The cost of correct assignments is set to zero, i.e., $\rho_n(0,0) = \rho_n(1,1) = 0$ and the total
cost $\rho\colon \Omega^N \times \Omega^N \to \mathbb{R}$ is given by summation $(\mathbf{x},\mathbf{y}) \mapsto \rho(\mathbf{x}, \mathbf{y}) = \sum_{n=1}^{N} \rho_n(x_n, y_n)$.

We now turn to the mathematical description of  testing strategies.
A \textit{testing strategy} for $N$ individuals consists of a test procedure and a decision procedure.
An 
$(N,K)$-\textit{test procedure}
is given by 
$\boldsymbol{\eta} = (\eta_1, \eta_2, \dots, \eta_K)$ where $\eta_1 \in \mathcal{P}\big(\{1,2,\dots,N\}\big)$ and for $k > 1$, $\eta_k\colon \Omega^{k-1} \to  \mathcal{P}\big(\{1,2,\dots,N\}\big)$, where $\mathcal{P}\big(\{1,2,\dots,N\}\big)$ denotes the collection of all subsets of $\{1,2,\dots,N\}$. 
Here, the set $\eta_1$ indicates the group used for the first test and the set-valued function $\eta_k$ indicates the group used for the $k$-th test, given the results of the previous $k-1$ tests. 
The corresponding \emph{test function} $\boldsymbol{\vartheta}\colon \Omega^N \to \Omega^K$ is given by $\vartheta_1(\mathbf{x}) = \max\{x_n | n \in \eta_1\}$ and for $k > 1$, we have $\vartheta_k(\mathbf{x}) = \max\big\{x_n \big| n \in \eta_{k}\big(\vartheta_1(\mathbf{x}), \vartheta_2(\mathbf{x}), \dots, \vartheta_{k-1}(\mathbf{x})\big)\big\}$.
Thus, the $k$-th component $\vartheta_k$  corresponds to the result of the $k$-th test.

A
$(K,N)$-\textit{decision procedure}
is a mapping 
$\kappa\colon \Omega^K \to \Omega^N$
that assigns, based on the outcome of $K$ tests, a status (infected/healthy) to all $N$ individuals. 

The concatenation of $\boldsymbol{\vartheta}$ and $\kappa$ maps the true status $\mathbf{X} := (X_1, \dots, X_N)$ of all $N$ individuals to an estimated status $(Y_1, \dots, Y_N) = \kappa \big(\boldsymbol{\vartheta} (X_1, \dots, X_N) \big)$ of these individuals using $K$ group tests.
In total, this corresponds to $R=K/N$ \ac{TpI} which is referred to as the \textit{rate} of the testing strategy.
If $R < 1$, which is the regime we are interested in, there is a positive probability that the tests will not enable a perfect identification of all  infected  individuals.
Note that such an estimate ($0$: not infected, $1$: infected) has to be given for all $N$ individuals. We do not allow for individuals to be ``skipped,'' which would correspond to a ternary output alphabet.

To assess the average cost of the wrong assignments for given test and decision procedures, we use the \emph{expected cost per individual}, defined as the expected value
\begin{equation}
    D_{\text{test}}^{(\mathbf{X}, \rho)}(\boldsymbol{\vartheta}, \kappa) = \mathbb{E}\bigg[\frac{1}{N}\sum_{n=1}^N \rho_n(X_n, Y_n)\bigg].
\end{equation}
Because there are only finitely many possible  choices for $\boldsymbol{\vartheta}$ and $\kappa$, we can define the minimum
\begin{equation}
\label{eq:defdfink}
  \dfink = \min_{\boldsymbol{\vartheta}, \kappa} D_{\text{test}}^{(\mathbf{X}, \rho)}(\boldsymbol{\vartheta}, \kappa) ,
\end{equation}
where $\boldsymbol{\vartheta}$ and $\kappa$ range over all $(N,K)$-test and $(K,N)$-decision procedures, respectively.
The quantity $\dfink$ specifies the minimal cost, measured by $\rho$, that can be achieved by using $K$ tests for the $N$ individuals with random infection status $\mathbf X$.

\section{Results}
\label{sec:results}
Calculating $\dfink$ directly from \cref{eq:defdfink} is computationally infeasible, unless $N$ and $K$ are very small.
Nevertheless, we can use information-theoretic ideas to provide bounds.
These bounds are based on the idea of keeping the rate $R = K/N$ fixed, while letting $N$ approach infinity.
Our first result is a lower bound, that holds if all individuals share the same parameters.
Here, and in the remainder of the paper, we use the symbol $H_2(p) = -p\log p - (1-p)\log(1-p)$ for the binary entropy function, $\log(\,\cdot\,)$ denotes the logarithm to base $2$, and we adopt the convention that ``$0\cdot\log 0 = 0$.'' The proofs of the results in this \lcnamecref{sec:results} are presented in \cref{sec:itdetails} and in the appendices.

The following \lcnamecref{thm:RD_function} presents a lower bound on $\dfink$ for the case when infection is equally likely and independent across the entire population. It follows immediately from the more general \cref{th:partitions}, which will be stated later in this \lcnamecref{sec:results}. 
\begin{theorem}
  \label{thm:RD_function}
  Let $\mathbf{X} := (X_1, \dots, X_N)$ be $N$ independent and identically distributed Bernoulli($p$) random variables describing the infection status of a given population. 
  The cost of wrong assignments is given by \cref{eq:false} with $b_n = b$ and $c_n = c$ for all $n$.
  Furthermore, set $a := c/b > 0$.
  For $v\in [0, v_0)$ define
  \begin{align}
      \label{eq:defdviidsup}
      \bar D(p, a, v)
      &= p \bigg(\frac{v}{1-v}-\frac{av^a}{1-v^a} \bigg)
            + \frac{a}{1-v^a} - \frac{a+v^{a+1}}{1- v^{a+1}} \\
      \label{eq:defrviidsup}
      \bar R(p, a, v)
      &= \bar D(p, a, v) \log v + H_2(p) - \log \bigg(\frac{1- v^{a+1}}{1-v^a}\bigg)
            + p \log \bigg(\frac{1-v}{1-v^a} \bigg)
  \end{align}
  and for $v \ge v_0$, define $\bar D(p,a,v) = \min \big\{1-p,a p \big\}$ and $\bar R(p,a,v) = 0$.
  Here, $v_0$ is the smallest solution $v > 0$ of the equation
  \begin{align}
      (p v^{a+1} + 1 - p - v)(p v^{-a-1} + 1 - p - v^{-1}) = 0 .
      \label{eq:solv0}
  \end{align}
  Then there cannot exist a testing strategy that uses fewer than $\bar R(p, a, v)$ \ac{TpI} (i.e., $\bar R(p, a, v) N$ tests in total) and achieves an expected cost less than $b \bar D(p, a, v)$, i.e., 
  $\dfink \geq  b \bar D(p, a, v)$ for all $K\leq \bar R(p, a, v) N$.
\end{theorem}
We emphasize that, in contrast to similar results in classical information theory, the lower bound in \cref{thm:RD_function} cannot always be achieved arbitrarily closely for increasing $N$.
For example, for $p = \frac{1}{2}(3-\sqrt{5}) \approx 0.381$ and $v=0$, we obtain $\bar D(p, a, v)=0$ and $\bar R(p, a, v)=H_2(p)\approx 0.959$, although it is known \citep{ungar1960cutoff} that only individual testing (i.e., $R=1$) can achieve $\dfink=0$ in this setting.

Even though the lower bound in \cref{thm:RD_function} is somewhat cumbersome and difficult to grasp intuitively, it can easily be calculated for a given scenario.  
We will compare it to proposed testing strategies in \cref{fig:rda50} to illustrate its applicability.

\begin{figure}[p!]
  \centering
  \begin{tikzpicture}
  \begin{axis}[
    DRaxis,
    DRaxisTpI,
        legend entries={Individual Testing, Binary Splitting, 1SG, 2SG, Lower Bound},
        xmin=0,xmax=0.11,
    ymin=0,ymax=0.51,
    max space between ticks=50pt,
    ]
    \addplot+[DRplot, blue, dashed] table [DRtableTpI] {data/individual_a50_p0.01.csv}; 
    \addplot+[DRplot, blue, solid] table [DRtableTpI] {data/binsplit_a50_p0.01.csv}; 
    \addplot+[DRplot, green!50!black, solid] table [DRtableTpI] {data/1SG_a50_p0.01.csv}; 
    \addplot+[DRplot, red, solid] table [DRtableTpI] {data/2SG_a50_p0.01.csv}; 
    \addplot+[DRplot, black] table [DRtableTpI] {data/RD_a50_p0.01.csv};
    \addplot+[DRlabeled, green!50!black, solid, only marks] table [DRtableTpI_label] {data/1SG_a50_p0.01_label.csv};
    \addplot+[DRlabeled, red, solid, only marks] table [DRtableTpI_label] {data/2SG_a50_p0.01_label.csv}; 
  \end{axis}
\end{tikzpicture}
  \caption{
    Testing of a population with probability of infection $p=0.01$ and cost parameters $b=1$ and $c=50$. The lower bound from \cref{thm:RD_function} is compared to the binary splitting algorithm, and our strategies 1SG and 2SG as well as the significantly worse individual testing. Markers with associated numbers indicate the testing strategy used to achieve the given point, e.g., $(32)$ indicates the use of 1SG($32$) and $(66,22)$ indicates the use of 2SG($66$,$22$).}
  \label{fig:rda50}
\end{figure}
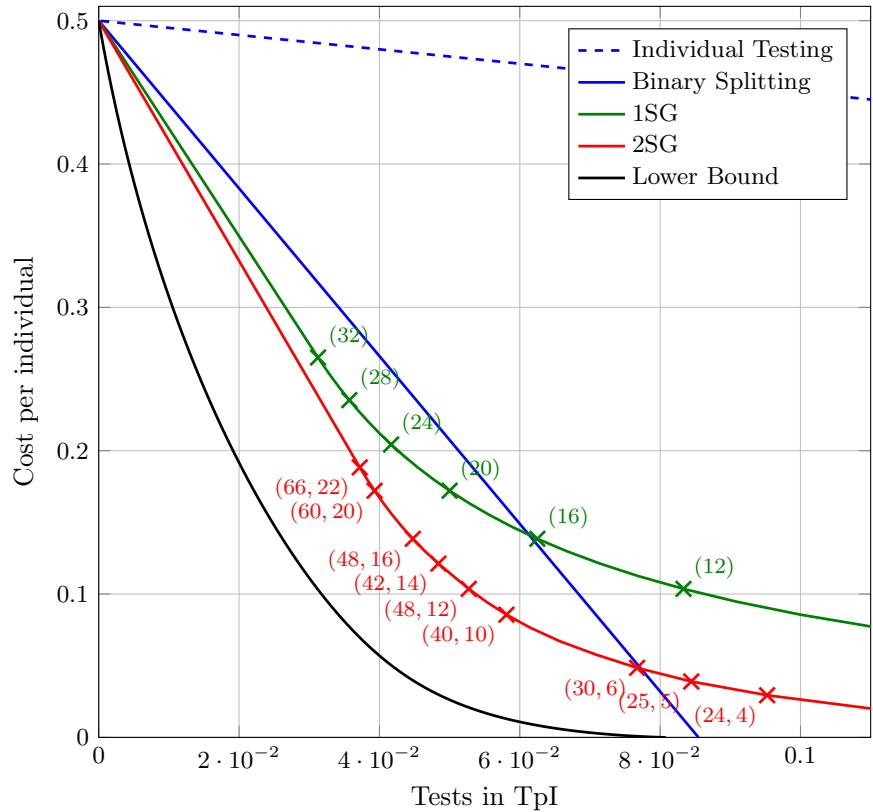

We next calculate the necessary number of tests and the resulting cost for some simple testing strategies.
We first consider the strategy, proposed by \citet{jonnerby2020maximising}, to 
separate the population into equally sized and disjoint groups, test each group,
and assign an infected status to each member of a positive group without conducting further tests.
We refer to this strategy as \textit{one stage group testing (1SG)} and designate the size of the group in parenthesis, e.g., 1SG($50$) for a group size of $50$ individuals. 
We restate the following simple result from Section~2.3.1 in \citet{jonnerby2020maximising}, which is a special case of the more general \cref{lem:ksg} below.
\begin{lemma}
    The 1SG($u$) testing strategy has a rate of $1/u$ \ac{TpI} and an expected cost of 
    $D_{\textrm{1SG($u$)}} = b \big(1 - p - (1-p)^u\big)$.
\end{lemma}
Note that the exact rate and expected cost can be achieved only for a population $N$ that is an integer multiple of the group size $u$. 
However, the overhead is at most $1$ additional test for a final group of smaller size, resulting in an additional $1/N$ \ac{TpI} and a negligible decrease of the expected cost. 
Since we are usually concerned with large $N$, we ignore these terms in this and the following testing strategies.

Of course, 1SG($u_1$) can be further extended by testing those individuals again, that belong to a positively tested group of size $u_1$.
Specifically, we can separate the group into disjoint subgroups of smaller size $u_2$ and subject these subgroups again to a group test.
This \textit{two stage group testing (2SG)} strategy is a generalization of Dorfman testing \cite{dorfman43} which corresponds to the case $u_2 = 1$, i.e., each individual in a positive group is tested separately. 
The decision strategy remains the same as in 1SG: Those individuals that belong to a positively tested subgroup (of size $u_2$) are declared infected.
We again designate the group sizes in parenthesis, e.g., 2SG($50,10$) for a group size of $50$ individuals that is divided into five subgroups of $10$ individuals each, if the first group test is positive.  
Again, we can calculate the rate and expected cost in closed form.
\begin{lemma}
    The 2SG($u_1,u_2$) testing strategy has an expected rate of 
            $R_{\textrm{2SG($u_1$,$u_2$)}} = \frac{1}{u_1} + \frac{1-(1-p)^{u_1}}{u_2}$
        and an expected cost per individual of 
    $D_{\textrm{2SG($u_1$,$u_2$)}} = b (1 - p - (1-p)^{u_2})$.
\end{lemma}

Evidently, more than 2 stages could be used, but even more bookkeeping is then required and, typically, very large group sizes are needed to obtain a benefit. Thus, this extension may not be practically useful anymore. 
For the sake of completeness, we nevertheless provide the rate and expected cost of $k$ stage group testing, abbreviated as $k$SG($u_1, \dots, u_k$), where $u_\ell$ denotes the group size at stage $\ell$.
Here, positive groups of size $u_\ell$ are  separated into smaller subgroups of size $u_{\ell+1}$ for $\ell=1, \dots, k-1$ and individuals that belong to a positive subgroup at stage $k$ are declared infected, while all others are declared healthy. 
A proof of the result is given in \cref{app:proofksg}.
\begin{lemma}
\label{lem:ksg}
    The $k$SG($u_1, \dots, u_k$) testing strategy has an expected rate of 
    \begin{equation}
    \label{eq:rateksg}
        R_{\textrm{$k$SG($u_1, \dots, u_k$)}} = \frac{1}{u_1} + \sum_{\ell =1}^{k-1}\frac{1-(1-p)^{u_\ell}}{u_{\ell+1}}
    \end{equation}
    and an expected cost per individual of 
    \begin{equation}
        \label{eq:distksg}
        D_{\textrm{$k$SG($u_1, \dots, u_k$)}} = b (1 - p - (1-p)^{u_k})\,.
    \end{equation}
\end{lemma}

\begin{example}
\label{ex:a50}
    We consider the setting $p=0.01$, $b=1$, and $c = a = 50$, i.e., we assume that we have a prevalence of $1\%$ and a false negative is fifty times worse than a false positive assignment. 
    We begin with some trivial observations.
    First, if we do not have any tests available, the best strategy is to assign everyone to be healthy. 
    This is because the expected cost of declaring an individual healthy is $\mathbb{E}\big[\rho(X_n, 0)\big] = p\cdot a + (1-p)\cdot  0 = 0.5$,
    while the expected cost for declaring someone infected is $\mathbb{E}\big[\rho(X_n, 1)\big] = p\cdot 0 + (1-p)\cdot 1 = 0.99$.
    On the other extreme is the case of zero cost, i.e., to determine exactly which individuals are infected.
    Here, clearly individual testing at a rate of $1$ \ac{TpI} could be applied, but Theorem~2 in \citet{aldridge2019rates} shows that it is also possible using a binary splitting algorithm at a rate of $0.0855$ \ac{TpI}.
    Any rate-cost tradeoff between these extreme points can be achieved by applying the available tests to as many individuals as possible, while declaring all others as healthy by default.
    Unfortunately, using approaches like a binary splitting algorithm do not come without problems: 
    There is a significant amount of bookkeeping required and individuals usually have to be tested many times in a row delaying the notification about the result.
    Shifting to the strategies 1SG and 2SG that do not aim at identifying the status of each individual but to minimize the overall expected cost can lead to simpler procedures and better performance at the same time.
    Our lower bound in \cref{thm:RD_function}, the binary splitting algorithm, and the strategies 1SG and 2SG, as well as individual testing, for comparison, are illustrated in \cref{fig:rda50}.
    Note, that for all strategies, one can always subject only part of the individuals to a test and simply declare the rest healthy. 
    In particular, we see that the optimal 2SG strategy when tests are scarce (less than $0.037$ \ac{TpI}) is to use the testing strategy 2SG($66, 22$) for as many individuals as possible, outperforming the binary splitting approach. 

\end{example}

The example above illustrates that the proposed strategies, although very simple, can outperform the best known testing strategies if there are not sufficient  tests available to test all individuals. 
This is because previous strategies always aimed at exactly identifying the infection status and did not consider the possibility to declare an infection status based on imperfect information.
However, there are also many situations where the simple strategies we discussed above are useless. 
In particular, if the relative cost $a$ is small, then it is hardly ever useful to declare an individual infected if one is not very sure about the infection status.

We now turn to the more general setting of a heterogeneous population, i.e., individuals may have different prevalence $p_n$ and also different costs $b_n$ and $c_n$.
One simple example of this situation is to distinguish between  individuals with symptoms (that have a higher prevalence $p_n$) and individuals without symptoms.
Additionally, different costs may occur for individuals in critical areas, e.g., health care.
There, a higher risk of infecting vulnerable individuals may increase the cost $c_n$ of false negatives, and additionally, a higher false positive cost $b_n$ might be incurred due to the importance of the work being performed by these individuals.
We will discuss a specific scenario in \cref{sec:case} but first present the lower bound, an extension of \cref{thm:RD_function}.
The proof of the following \lcnamecref{th:partitions} will be presented at the end of \cref{sec:itdetails}, as it requires several results from rate-distortion theory which are introduced there.
\begin{theorem}
\label{th:partitions}
  Let $\mathbf{X} := (X_1, \dots, X_N)$ be $N$ independent Bernoulli random variables describing the infection status of a given population. 
  Assume further that the total population is separated into $I$ subpopulations, of sizes $N^{(1)}, \dots, N^{(I)}$, and that all $N^{(i)}$ individuals in the $i$-th subpopulation have the same probability $p^{(i)}$ of infection and are measured using the same cost function with parameters $b^{(i)}, c^{(i)}$ as in \cref{eq:false}. Define $a^{(i)} = c^{(i)}/b^{(i)} > 0$ and for $v \in [0,1]$, let 
  \begin{align}
    \label{eq:subpopDR}
    D(v) = \frac{1}{N} \sum_{i = 1}^I N^{(i)} b^{(i)}  \bar D\big(p^{(i)}, a^{(i)}, v^{b^{(i)}}\big), \;
                R(v) = \frac{1}{N} \sum_{i = 1}^I N^{(i)} \bar R\big(p^{(i)}, a^{(i)}, v^{b^{{(i)}}}\big),
  \end{align}
  where  $\bar D(p, a, v)$ and $\bar R(p, a, v)$ are defined in \cref{thm:RD_function}.
  Then there cannot exist a testing strategy that uses fewer than $R(v)$ \ac{TpI} (i.e., $R(v) N$ tests in total) and achieves an expected cost less than $D(v)$, i.e., $\dfink \geq D(v)$ for all $K\leq R(v) N$.
\end{theorem}

\section{Historic example case}\label{sec:case}

We consider the SARS-CoV2 pandemic situation in Austria in mid November 2020.
On average there were $N=8\,916\,845$ individuals living in Austria in 2020~\cite{StatistikAustria2020Population}.
During the three days from 12th of November 2020 until 14th of November 2020 a total of $N_t = 103\,621$ individuals were subjected to a PCR-test for SARS-CoV2~\cite{OAGEG2021Dashboard}.
Of these $N_t$ individuals, $20\,349$ tested positive~\cite{OAGEG2021Dashboard}.
This corresponds to a prevalence of $p_t = 0.196$ in this tested subpopulation. 
We will refer to this subpopulation as the \textit{high prevalence subpopulation}.
At the same time, a prevalence study found that approximately $3.1\%$ of the total population were  infected~\cite{STATISTIKAUSTRIA2021COVID}.
Thus, in the untested population of $N_u = 
8\,813\,224$ the prevalence was about $p_u = 0.029$.
We will refer to this subpopulation as the \textit{low prevalence subpopulation}.

To illustrate the full potential of our model, we further consider individuals working in health care separately and will assign a higher cost for wrong assignments within this subpopulation. 
The most recent count of health care professionals working in hospitals and health care centers in Austria was 
$N_h = 121\,567$ at the end of 2019~\cite{STATISTIKAUSTRIA2021Personal}.
Note that a finer separation into subpopulations is of course possible but avoided for the sake of simplicity.

It is difficult to argue for the choice of specific costs of wrong assignments. 
However, in mid November 2020, a lockdown was issued in Austria,\footnote{\url{https://www.ris.bka.gv.at/Dokumente/BgblAuth/BGBLA_2020_II_479/BGBLA_2020_II_479.pdfsig} (in German)} which we interpret as the turning point where considering all (untested) individuals as infected is less costly than considering these individuals as healthy. 
This implies $(1-p) b \approx p c$, thus, at a prevalence of $p = 0.029$ in the untested, general population,  we obtain $c \approx 33 b$. 
We normalize $b=1$ and choose $c = 33$.
Since health-care facilities were not closed during this time, we assume that individuals working in health care have a different trade-off between $b$ and $c$.
Specifically, we assume that a false positive assignment for this subpopulation is significantly more expensive and we set $b=6$ for these individuals. 
Although also a different $c$ value could be argued, we assume that the professional training counters the higher risk due to closer contact with susceptible individuals and we keep $c$ the same.

Since there is no data available on the prevalence within the health care system, we assume that it is the same as in the general population. In particular, also the split into the high prevalence and low prevalence subpopulations is assumed to be the same.
We thus end up with the following four subpopulations
\begin{itemize}
        \item The first subpopulation consists of individuals working in health care that belong to the high prevalence subpopulation with $p^{(1)}=0.196$ 
        and costs $b^{(1)}=6$, $c^{(1)} = 33$. 
            We assume that $N^{(1)} = 1\,413$ belong to this subpopulation.
        \item The second subpopulation consists of individuals working in health care that belong to the low prevalence subpopulation with $p^{(2)}=0.029$ 
        and costs $b^{(2)}=6$, $c^{(2)} = 33$. 
            We assume that $N^{(2)} = 120\,154$ belong to this subpopulation.
        \item The third subpopulation consists of individuals not working in health care that belong to the high prevalence subpopulation with $p^{(3)}=0.196$ 
        and costs $b^{(3)}=1$, $c^{(3)} = 33$. 
            We assume that $N^{(3)} = 102\,208$ belong to this subpopulation.
        \item The fourth subpopulation consists of individuals not working in health care that belong to the low prevalence subpopulation with $p^{(4)}=0.029$ 
        and costs $b^{(4)}=1$, $c^{(4)} = 33$. 
        We assume that $N^{(4)} = 8\,693\,070$ belong to this subpopulation.
\end{itemize}

\begin{figure}[tbh]
  \centering
  \begin{tikzpicture}
    \begin{axis}[
      DRaxis,
      legend entries={Individual Testing, 1SG/2SG, Binary Splitting, Lower Bound},
            ymin=0.0,ymax=1.0,
      xmin=0.0, xmax=2.5e4,
      max space between ticks=45pt,
      ]
      \addplot+[DRplot, red, dotted] table[DRtable] {data/individual_p0.20_0.03_0.20_0.03.csv};
      \addplot+[DRplot, blue, solid] table [DRtable] {data/2SG_p0.20_0.03_0.20_0.03.csv};
      \addplot+[DRplot, green!50!black, solid] table [DRtable] {data/binsplit_p0.20_0.03_0.20_0.03.csv};
      \addplot+[DRplot, black] table[DRtable] {data/RD_p0.20_0.03_0.20_0.03.csv}; 
      \addplot+[DRlabeled, blue, only marks] table [DRtable_label] {data/2SG_p0.20_0.03_0.20_0.03_labels.csv};
    \end{axis}
  \end{tikzpicture}
  \caption{Testing of the Austrian population grouped into 4 different sub-populations in mid November 2020.
       The lower bound from \cref{th:partitions} is compared to the binary splitting algorithm, and an optimal combination of strategies 1SG and 2SG as well as the significantly worse individual testing. Markers with associated numbers indicate the four testing strategies used to achieve the given point, where $\infty$ indicates that the given subpopulation is not tested, e.g., $(2)(\infty)(\infty)(30,10)$ indicates that subpopulations $2$ and $3$ are not tested, subpopulation $1$ is tested using  1SG($2$), and subpopulation $4$ is tested using  2SG($30$,$10$).}
  \label{fig:subpop}
\end{figure} 

Our lower bound and the performance of the optimal combinations of 1SG and 2SG strategies for these specific parameters are illustrated in \cref{fig:subpop}
in comparison to  optimized binary splitting algorithms and individual testing.
More specifically, the depicted curves correspond to the optimal use of the available tests under the given strategy.
For individual testing and the binary splitting algorithms, this means that depending on the number of tests, a certain amount of people can be tested and correctly informed about their status, 
while many untested individuals remain and are assigned a status without being tested resulting in the depicted cost.

For the $103\,621$ tests used during the discussed period, which correspond to $0.0116$ \ac{TpI}, the optimized strategy is to subject as many individuals from the fourth subpopulation as possible to a test using a 1SG($33$) testing strategy. 
By this approach, the expected cost is $0.816$. 
Individual testing could only reduce the expected cost to $0.944$, hardly any improvement over not testing at all at an expected cost of $0.956$.
Note that in all subpopulations and strategies we assume that an optimal decision for untested individuals is used, namely, that a healthy status is assigned to untested individuals in subpopulations two and four and an infected status to untested individuals in subpopulations one and three.

To illustrate the implications of this kind of strategy, we calculate the expected number of individuals that are considered infected.
First, this number includes all individuals in the subpopulations one and three.
Furthermore, each 1SG($33$) test in subpopulation four has a probability of $1-(1-0.029)^{33} \approx 0.621$ to result in a positive outcome. 
Thus, an additional expected number of  
$2\,124\,712$ individuals would be considered infected in this subpopulation.
In total, an expected number of $2\,228\,333$ individuals would be considered infected.

We see that the decrease in expected cost is very limited with such a low number of tests, even if an optimized strategy is used.
In particular, our lower bound shows that no strategy can reduce the expected cost below $0.609$ using the assumed $103\,621$ tests. 
However, our theory can also be used to estimate how many tests are required to achieve a certain reduction in expected cost. 
For example, we can find the number of tests necessary to obtain half the expected cost than without testing, i.e., $0.478$.
Our lower bound shows that at least $0.0226$  \ac{TpI} ($201\,256$ tests in total) are necessary, while our optimized strategies show that $0.0419$  \ac{TpI} ($373\,636$ tests in total) are sufficient.
Our optimized strategies suggest to use a 1SG($4$) strategy to test all individuals in subpopulation one and a mix of 1SG($24$) and  1SG($23$) strategies to test all individuals in  subpopulation four.
We also want to point out that an astonishing number of $0.499$  \ac{TpI} ($4\,447\,461$ tests in total) would be required to achieve the same goal with individual testing;
and even the previously best known binary splitting testing strategies require $0.102$ \ac{TpI} ($909\,637$ tests in total).

A second example case is presented in~\cref{sec:second-example-case}, analyzing the situation in Austria during the SARS-CoV-2 pandemic in early April 2020.
The two examples showcase how our theory can be used to explore possible strategies for a given set of parameters and a given number of tests. 
Furthermore, they already illustrate that the specific scenario can affect all parts of the optimal testing strategy, e.g., there is no general rule that the testing of individuals with highest prevalence has priority. 
Maybe surprisingly, the examples also illustrate that there are situations where our approach can result in substantial gains using very simple testing strategies. 
Any other parameter choices can be explored using the Jupyter notebook and Python code, available at \url{https://github.com/g-pichler/group-testing}. This software was created using a Jupyter Notebook~\citep{Perez2007IPython}, as well as SciPy~\citep{Virtanen2020SciPy}, NumPy~\citep{Harris2020Array}, and Matplotlib~\citep{Hunter2007Matplotlib}.

Finally, we want to emphasize, that the historic example cases discussed here certainly do not completely reflect the reality of testing in Austria in 2020. In addition to the somewhat arbitrary choices of values $b, c$, our analysis suffers from the shortcomings that are to be discussed in~\cref{sec:discussion}. E.g., considering the first example case~(\cref{fig:subpop}), a 1SG($33$) test is positive with probability around $0.621$, while the chance of infection is less than $0.03$. One would not expect this strategy to be viable, as compliance will be low.

\section{Information theoretic details}
\label{sec:itdetails}

Readers familiar with information theory will see the obvious analogues between our problem formulation and classical rate-distortion theory \cite{gray90source}.
Indeed, an $(N,K)$-test function can be seen as a specific encoder of a binary sequence, whereas the corresponding $(K,N)$-decision procedure is a specific decoder.
Thus, $\dfink$ defined in \cref{eq:defdfink} is the minimal expected distortion of a restricted class of source codes of rate $R = K/N$ for the vector $\mathbf{X}$.

A fundamental information-theoretic result provides a lower bound on the minimal expected distortion of an arbitrary source code of rate $R$, the so-called information distortion-rate function.
\begin{definition}
  \label{def:DR_function}
  Let $\mathbf{X}$ be a random variable on $\Omega^N$ and $\rho\colon \Omega^N \times \Omega^N \to \mathbb R_+$ a distortion function.
  The \emph{information distortion-rate function} of $(\mathbf{X}, \rho)$ is given by
  \begin{align}
    \label{eq:DR_function}
    D_{\textrm{I}}^{(\mathbf{X}, \rho)}(R) 
    = \min_{p_{\mathbf{Y}|\mathbf{X}} : I(\mathbf{X}; \mathbf{Y}) \le NR} \mathbb E_{p_{\mathbf{X}} p_{\mathbf{Y}|\mathbf{X}}} \bigg[ \frac{1}{N} \rho(\mathbf{X}, \mathbf{Y}) \bigg],
  \end{align}
  where $I({}\cdot{};{}\cdot{})$ denotes mutual information in bits \cite[Sec.~2.6]{gray90source}.
\end{definition}
Because $D_{\textrm{I}}^{(\mathbf{X}, \rho)}(R)$ is a lower bound on the minimal expected distortion of an arbitrary source code of rate $R$, it can be used to obtain a lower bound on $\dfink$. The following \lcnamecref{lem:rdgtrd} is a direct consequence of Theorem~3.2.1 in \cite{gray90source} and a proof is thus omitted.
\begin{lemma}
\label{lem:rdgtrd}
            The minimal cost $\dfink$, as defined in \cref{eq:defdfink}, is lower-bounded by the information distortion-rate function $D_{\textrm{I}}^{(\mathbf{X}, \rho)}(R)$, i.e., for all $R\geq 0$, 
    $D^{(\mathbf{X}, \rho)}\big(K) 
        \geq D_{\textrm{I}}^{(\mathbf{X}, \rho)}(R)$,
    when $K \le RN$.
    In particular, no testing strategy can achieve an expected cost strictly less than $D_{\textrm{I}}^{(\mathbf{X}, \rho)}(R)$ if only $R$ \ac{TpI} are available.
\end{lemma}
To prove our main results, it remains to characterize the information distortion-rate function $D_{\textrm{I}}^{(\mathbf{X}, \rho)}(R)$. We first focus on the case $N=1$ and will subsequently extend our setting to prove \cref{th:partitions}.
The following result is based in the variational description of the information distortion-rate function (see Section 2.4 in \citet{gray90source}) and a detailed proof is provided in \cref{app:proofiid_rd_final}.

 \begin{theorem}
   \label{thm:iid_rd_final}
  Let $X\in \Omega$ be a Bernoulli($p$) random variable 
    and $\rho\colon \Omega^2\to [0,\infty)$ a distortion function satisfying $\rho(0,0)=\rho(1,1)=0$, $\rho(0,1)=1$, and $\rho(1,0)=a$.
  Then the entire information distortion-rate function is parameterized as $D_{\textrm{I}}^{(X, \rho)}\big(\bar R(p, a, v)\big) = \bar D(p, a, v)$ with $v \in [0, 1]$.
  Here, $\bar D(p, a, v)$ and $\bar R(p, a, v)$ are defined in \cref{thm:RD_function}.
\end{theorem}

\begin{remark}
  \label{rmk:DR_deriv}
  For later use, we note that under the assumptions of \cref{thm:iid_rd_final}, $1/ (\log v)$ is the slope of the information distortion-rate function \cite[p.~86]{gray90source} for $v \in [0,v_0]$, i.e., 
  $\dot D_{\textrm{I}}^{(X, \rho)}\big(R(p,a,v)\big) = (\log v)^{-1}$.
  Furthermore, note that $\big(\bar R(p, a, v_0), \bar D(p, a, v_0)\big)$ is the point $(0, \min\{ap,1-p\})$ on the distortion-rate curve. Thus, in particular, $\dot D_{\textrm{I}}^{(X, \rho)}(0) = (\log v_0)^{-1}$.
\end{remark}

We can now state the central characterization result of the information distortion-rate function in the setting of \cref{th:partitions}.
A detailed proof of the result is provided in \cref{app:proofsubpop2}.

\begin{theorem}
  \label{th:subpop2}
  Let $\mathbf{X}$ be as in \cref{th:partitions}.
  Then $D_{\textrm{I}}^{(\mathbf{X}, \mathbf{\rho})}\big(R(v)\big) = D(v)$ for every $v \in [0,1]$,
  where $D(v)$ and $R(v)$ are given in \cref{eq:subpopDR}.
\end{theorem}

Combining the previous results, we can now prove \cref{th:partitions,thm:RD_function}.
\begin{proof}[Proof of~\cref{th:partitions,thm:RD_function}]
  By~\cref{lem:rdgtrd}, we have for all $R\geq 0$, that $\dfink \geq D_{\textrm{I}}^{(\mathbf{X}, \rho)}(R)$ if $K \le NR$. Using the characterization of the information distortion-rate function in~\cref{th:subpop2}, this yields
  $\dfink \geq D(\nu)$ whenever $K \le NR(\nu)$, concluding the proof of~\cref{th:partitions}.

  \Cref{thm:RD_function} is merely a special case of \cref{th:partitions} with $I=1$.
\end{proof}

Finally, we want to point out that many information theoretic questions about the problem remain open. 
In particular, the characterization of a ``group testing distortion-rate function'' that can be defined as the infimum of all expected distortions that are achievable at a given rate as the population size $N$ goes to infinity is a challenging new problem that can also be formulated as a problem in classical rate-distortion theory with significant restrictions on the encoder. 
Here, non-adaptive group testing may be easier to analyze because the encoder is essentially restricted to the $\max(\,\cdot\,)$ function.
Another direction for future work is to incorporate testing errors, resulting in a joint source-channel coding scenario. 
It is not clear whether the source--channel separation theorem holds in this case. Nevertheless, we expect that a lower bound similar to the one presented here can be obtained, by adding the additional bits (i.e., tests) necessary for error-free communication over a binary (most likely asymmetric) channel.

\section{Discussion}
\label{sec:discussion}
We introduced a rigorous mathematical formulation of optimal testing strategies for a given number of available group tests. 
The problem is formulated as a one-time testing procedure in the sense that we assume that the infection status of the population does not change during the testing procedure and we do not consider any time evolution.
This enables us to use only very few parameters that describe a current pandemic situation and we do not require any detailed personal information such as contact maps between individuals.
However, even these few parameters can be difficult to estimate and may also result in ethical challenges.
The prevalence $p$ is the easiest and most obvious parameter and can be estimated using a negligible number of tests for a pilot study and the estimate can be improved as the testing strategy is applied.
The costs $b$ and $c$ of wrong assignments are more difficult to choose.
They can be adapted to the infectiousness of the disease, the cost and effectiveness of quarantine, non-pharmaceutical interventions that reduce the risk of infection, and most likely many more variables. 
How to exactly choose these costs is beyond the scope of our work and very specific for a given situation.
It may also include political decisions by weighing health factors (e.g., minimize the number of infections by quarantining many individuals) against economic factors (e.g., minimize the number of quarantined individuals). 
Note however that at least implicitly these costs are used in political decisions: A society-wide lockdown can be interpreted as the assertion, that assigning all (untested) individuals an infected status is less costly than considering them to be healthy. 
Thus, at a prevalence $p$, this implies $(1-p) b \le p c$. 

Once the parameters are fixed for a given scenario, our theory on the one hand gives ultimate bounds on how large a cost reduction can be achieved by group testing, and on the other hand suggests optimal allocation of resources for simple (suboptimal) testing strategies.

The present work is only a first step towards establishing a mathematical theory of group testing.
Thus, we do not consider the most general scenario and focus on basic scenarios that are simplified in many aspects. 
In particular, the scope of our results is limited by the following assumptions: 
\begin{itemize}
\item We assume that tests are perfect, i.e., a group test of $u$ individuals is negative if and only if all $u$ individuals are healthy. However, we expect that a small error probability will not significantly influence the results and simple simulations incorporating these errors can be used to check the robustness in a specific scenario. A rigorous theoretical treatment on the other hand should not only consider a fixed error probability for a group test, but the  error probability should rather depend on the group size and the number of infected individuals within the group (cf.~\citet{pilcher2020group}). 
Thus, an extension of our work in this direction would imply a more quantitative approach of testing outcomes and not merely the binary options we consider here. 
\item In this work, no upper bound on the group size is assumed a priori. The information theoretic lower bounds would not be affected by such a limitation, but certain points of the $k$SG strategy will become infeasible, requiring minor modifications to the published code. However, a more thorough analysis would not impose a hard limit on the group size, but consider the trade-off between test accuracy and group size, as mentioned in the previous point.
\item We assume that the infection status of different individuals is independent.
  \item We assume that the choice of testing strategy does not alter the cost of wrong assignment. This assumption may be violated in reality, as e.g., individuals that are assigned an ``infected'' will likely show different levels of compliance, depending on whether they were tested individually, in a group, or not at all.
  Going beyond this assumption could be achieved by changing the infection status from a binary decision to a probability assignment and using a suitably adapted cost function.
\item The cost function is fixed and applied independently to each individual. E.g., we cannot express the fact, that quarantining a small number of individuals working in critical areas is hardly problematic but once a critical threshold is passed the cost of quarantining further individuals becomes more costly than the  linear increase assumed in our model.
\item We do not consider the collection of samples from individuals as a limiting factor, but merely the number of  tests is limited.
  \item The testing strategies presented in this work are not at all optimal but merely illustrate the potential of our approach. For example, we do not consider testing individuals in several groups at the first stage (as, e.g., in the array testing approach by~\citet{phatarfod1994use}), nor mixing individuals from different subpopulations within a group. 
  Nevertheless, our suggested strategies are surprisingly simple and can outperform even the best (significantly more complicated) strategies currently known.
\item We consider testing as a stand-alone task and do not incorporate it into a larger disease model, as is done, e.g., in \citet{berger2020seir}. In particular, the probabilities of infection and the cost functions are assumed to be fixed and do not change over time.
\end{itemize}

From a theoretical viewpoint, our analysis reveals several surprising insights that are in contrast to long established fundamental statements in group testing theory.
First, an established statement is that group testing is beneficial exactly in the regime $p< \frac{1}{2}(3-\sqrt{5}) \approx 0.381$ \cite{ungar1960cutoff}, i.e., in this regime, group testing cannot outperform individual testing.
However, this result was proven under the assumption that perfect identification of the infection status of each individual is required. 
Maybe surprisingly, this statement does not extend to our setting.
Indeed, for the specific case of $p= \frac{1}{2}(3-\sqrt{5})$, $b=1$, and $c=10$, using the testing strategy 1SG($2$) has a  strictly lower expected cost than individually testing every second individual (both have rate $1/2$ \ac{TpI}). 
Similarly, in the non-adaptive regime (i.e., all tests are performed in parallel) an established result is that individual testing is optimal if all individuals are independent and have a fixed prevalence $p$ \cite{aldridge2018individual}. 
Again, this result only holds if perfect identification of the infection status of each individual is required.
In our setting, the simple 1SG testing strategy is a non-adaptive strategy and is clearly superior to individual testing (see \cref{fig:rda50}).

\begin{appendix}

\section{Second example case}
\label{sec:second-example-case}

As a second example with significantly different prevalences, we consider the SARS-CoV2 pandemic situation in Austria in early April 2020.
More specifically, another prevalence study was conducted from 1st of April until 6th of April 2020. 
Here, a prevalence of approximately $p = 0.0033$ was found~\cite[Section~3.2]{ogob20}.
Since most of the samples were collected on the 4th of April, we consider a timeframe from 3rd of April until 5th of April. 
On these three days a total of $N_t =  16\,226$ tests were conducted~\cite{OAGEG2021Dashboard}.
Of those, $778$ did yield a positive result~\cite{OAGEG2021Dashboard}.
Thus, the prevalence in this high prevalence subpopulation was $p_t = 0.048$.
The resulting untested population is $N_u =  8\,900\,619$ with a prevalence of $p_u = 0.0032$.
We consider the same separation into subpopulations as before as well as the same costs.
\begin{itemize}
\item The first subpopulation consists of individuals working in health care that belong to the high prevalence subpopulation with $p^{(1)}=0.048$ 
  and costs $b^{(1)}=6$, $c^{(1)} = 33$. 
  We assume that $N^{(1)} = 221$ belong to this subpopulation.
\item The second subpopulation consists of individuals working in health care that belong to the low prevalence subpopulation with $p^{(2)}=0.0032$ 
  and costs $b^{(2)}=6$, $c^{(2)} = 33$. 
  We assume that $N^{(2)} = 121\,346$ belong to this subpopulation.
\item The third subpopulation consists of individuals not working in health care that belong to the high prevalence subpopulation with $p^{(3)}=0.048$ 
  and costs $b^{(3)}=1$, $c^{(3)} = 33$. 
  We assume that $N^{(3)} = 16\,005$ belong to this subpopulation.
\item The fourth subpopulation consists of individuals not working in health care that belong to the low prevalence subpopulation with $p^{(4)}=0.0032$ 
  and costs $b^{(4)}=1$, $c^{(4)} = 33$. 
  We assume that $N^{(4)} = 8\,779\,273$ belong to this subpopulation.
\end{itemize}

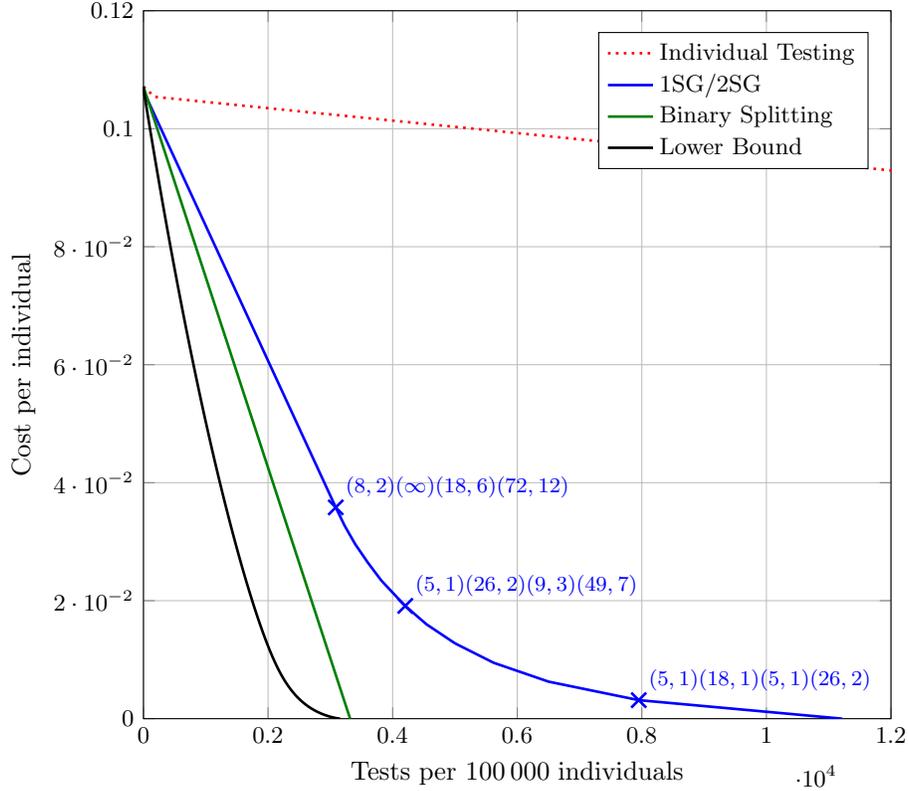
\begin{figure}[tbh]
  \centering
  \begin{tikzpicture}
    \begin{axis}[
      DRaxis,
      legend entries={Individual Testing, 1SG/2SG, Binary Splitting, Lower Bound},
            ymin=0.0,ymax=0.12,
      xmin=0.0, xmax=1.2e4,
      max space between ticks=45pt,
      ]
      \addplot+[DRplot, red, dotted] table[DRtable] {data/individual_p0.05_0.00_0.05_0.00.csv};
      \addplot+[DRplot, blue, solid] table [DRtable] {data/2SG_p0.05_0.00_0.05_0.00.csv};
      \addplot+[DRplot, green!50!black, solid] table [DRtable] {data/binsplit_p0.05_0.00_0.05_0.00.csv};
      \addplot+[DRplot, black] table[DRtable] {data/RD_p0.05_0.00_0.05_0.00.csv}; 
      \addplot+[DRlabeled, blue, only marks] table [DRtable_label] {data/2SG_p0.05_0.00_0.05_0.00_labels.csv};
    \end{axis}
  \end{tikzpicture}
  \caption{Testing of the Austrian population grouped into 4 different sub-populations in early April 2020.
      The lower bound from \cref{th:partitions} is compared to the binary splitting algorithm, and an optimal combination of strategies 1SG and 2SG as well as the significantly worse individual testing. Markers with associated numbers indicate the four testing strategies used to achieve the given point, where $\infty$ indicates that the given subpopulation is not tested, e.g., $(8, 2)(\infty)(18, 6)(72, 12)$ indicates that subpopulation $2$ is not tested and subpopulations $1$, $3$, and $4$ are tested using 2SG($8, 2$), 2SG($18, 6$), and 2SG($72, 12$), respectively.}
  \label{fig:subpop2}
\end{figure} 
Our lower bound and the performance of the optimal combinations of 1SG and 2SG strategies for these specific parameters are illustrated in \cref{fig:subpop2}.
For the $16\,226$ tests used during the discussed period, which correspond to $0.0018$ \ac{TpI}, the optimized strategy is to subject all individuals in subpopulation one to a test using a 2SG($8,2$) testing strategy,  all individuals in subpopulation three using a 2SG($18, 6$) testing strategy, and
as many individuals from the fourth subpopulation as possible using a 2SG($72,12$) testing strategy. 
Using this approach yields an expected cost of $0.1023$.
Individual testing could only reduce the expected cost to $0.1054$, while not testing at all results in an expected cost of $0.1072$.

\section{Proofs of information theoretic results}
\label{app:proofsit}

In this \lcnamecref{app:proofsit}, we provide detailed proofs of the two main information theoretic results in the main manuscript, namely \cref{thm:iid_rd_final} and \cref{th:subpop2}.

\subsection{Proof of Theorem~\ref{thm:iid_rd_final}}
\label{app:proofiid_rd_final}
 
The following variational description of the information distortion-rate function
 is a particularization of Corollary~4.2.1 in \citet{gray90source} to our setting of binary random variables with an asymmetric distortion function, substituting $v= 2^s$.
\begin{corollary}
  \label{cor:parrdfunc}
  Let $X$ be a Bernoulli($p$) random variable and $\rho\colon \Omega^2\to [0,\infty)$ a distortion function satisfying $\rho(0,0)=\rho(1,1)=0$, $\rho(0,1)=1$, and $\rho(1,0)=a$.
  Then for every $v \in [0,1]$ a point $(R_v, D_v)$ on the graph of the information distortion-rate function, i.e., $D_{\textrm{I}}^{(X, \rho)}(R_v) = D_v$, is given as
  \begin{equation}
    R_v = 
    D_v \log v +
    \min_{q\in [0,1]}
    \Big(- p \log \big(q +(1-q)v^{a}\big)
    - (1-p) \log \big(q v + (1-q)\big)
    \Big)  
    \label{eq:defrstd}
  \end{equation}
  and
  \begin{equation}
    D_v =
    \frac{ a p (1-q^*) \big(q^* v + (1- q^*)\big) v^{a} + q^*( 1- p) \big(q^* + (1- q^*) v^{a} \big) v}{\big(q^* +  (1 - q^*) v^{a}\big) \big(q^* v + (1- q^*)\big)}
    \label{eq:defdvprelim}
  \end{equation}
  where $q^*$ solves the minimization in \cref{eq:defrstd}.
\end{corollary}

The minimization in \cref{eq:defrstd} can easily be done exactly. 
Although one could use existing results, e.g., Theorem~4.2.3 in \citet{gray90source}, it is easier to derive the optimization directly than to particularize these general results.
\begin{lemma}
\label{lem:opimqs}
    For $v,p \in (0,1)$ and $a > 0$, the function $\phi\colon [0, 1] \to \mathbb{R}$ defined by
    \begin{equation}
        \phi(q)
        = - p \log \big(q +(1-q)v^{a}\big)
        - (1-p) \log \big(q v + (1-q)\big) 
        \label{eq:objminq}
    \end{equation}
    is convex and
    \begin{equation}
        \argmin_{q\in [0,1]} \phi(q) 
        = 
        \begin{cases}
            q_v & \textrm{ if } q_v \in (0,1) \\
            0 & \textrm{ if } q_v \leq 0 \\
            1 & \textrm{ if } q_v \geq 1\,,
        \end{cases}
        \label{eq:optqcases}
    \end{equation}
    where
    \begin{align}
       q_v := 
       \frac{ (1 - p) v^{a+1} +p - v^{a} }{ (1 -v) (1- v^{a})}\,.
       \label{eq:defqs}
    \end{align}
\end{lemma}
    
\begin{proof}
    The derivative of $\phi$ is given as
    \begin{equation}
        \phi'(q)
        = - p \frac{1-v^{a}}{q +(1-q)v^{a}}
        - (1-p) \frac{v-1}{q v + (1-q)}\,.
        \label{eq:minqderiv}
    \end{equation}
    Thus, $\phi'(q)=0$ if and only if
    \begin{equation}
        0
        = - p (q(v-1)+1)(1-v^{a})
        - (1-p) (q(1-v^{a}) +v^{a})(v-1)
    \end{equation}
    which is equivalent to
    \begin{equation}
         p(1-v^{a})
        +(1-p) v^{a}(v-1)
        = q(1-v)(1-v^{a})
    \end{equation}
    and in turn 
    to $q=q_v$.
    The second derivative is given as
    \begin{equation}
        \phi''(q)
        =  p \frac{(1-v^{a})^2}{(q +(1-q)v^{a})^2}
        + (1-p) \frac{(v-1)^2}{(q v + (1-q))^2}
        \label{eq:minqderiv2}
    \end{equation}
    which is  positive and thus $\phi$ is convex.
    Thus, if the critical value $q_v$ is in $(0,1)$ then it is the global minimum. 
    Otherwise, $q_v \geq 1$ implies $\phi'(q)<0$ on $(0,1)$ and thus the global minimum is at $1$, and 
    $q_v \leq 0$ implies $\phi'(q)>0$ on $(0,1)$ and thus the global minimum is at $0$.
\end{proof}

To prove \cref{thm:iid_rd_final}, it remains to combine \cref{cor:parrdfunc,lem:opimqs} and to note that, by \cite[Theorem~4.2.1b]{gray90source}, the entire distortion-rate curve is parameterized by $v$.
  
    We first show that $q_v \in (0,1)$ (i.e., the first case in \cref{eq:optqcases}) if and only if $v \in (0, v_0)$.
    By the definition of $q_v$ in \cref{eq:defqs}, $q_v > 0$ is equivalent to
    \begin{equation}
       (1 - p) v^{a+1} +p - v^{a} > 0
    \end{equation}
    and $q_v < 1$ is equivalent to
    \begin{equation}
       (1 - p) v^{a+1} +p - v^{a} <  (1 -v) (1- v^{a})\,.
    \end{equation}
    These equations are, in turn, easily seen to be equivalent to 
    \begin{equation}
       p v^{-(a+1)} + (1 - p) - v^{-1} > 0
       \label{eq:cond1qv}
    \end{equation}
    and
    \begin{equation}
       p v^{a+1} + (1 - p) - v  > 0\,.
       \label{eq:cond2qv}
    \end{equation}
    The function  $f(u) := p u^{a+1} + (1 - p) - u$ has derivatives
    $f'(u) = (a+1) p u^{a} - 1$
    and $f''(u) = (a+1) a u^{a-1}$ and is thus convex for $u > 0$.
    Furthermore, $f(1)=0$, $f(0) = 1-p > 0$, and $\lim_{u\to \infty}f(u) = \infty$ and hence the function has either one more zero or is nonnegative everywhere.
    In either case, \cref{eq:cond1qv,eq:cond2qv} are satisfied if and only if $v \in (0,v_0)$.
    Thus, the assumptions of \cref{thm:iid_rd_final} imply that the minimum in \cref{eq:defrstd} is $q_v$ and we can insert it into $D_v$ and $R_v$ in \cref{cor:parrdfunc}.
    To this end, we first derive $1-q_v$ as
    \begin{align}
       1- q_v 
                   & = 
       \frac{p v^{a+1} + (1 - p) - v}{ (1 - v) (1 - v^{a})}.
    \end{align}
    Next, we note that
    \begin{align}
        q_v + (1 - q_v)v^a
                                                        & = 
        \frac{ p  (1 - v^{a+1}) }{1 - v}
    \end{align}
    and
    \begin{align}
        q_v v + (1 - q_v)
                                                        & = 
        \frac{ (1 - p)  (1 - v^{a+1}) }{1 - v^{a}}\,.
    \end{align}
    Inserting these relations into $D_v$ in \cref{eq:defdvprelim}, simple algebraic manipulations yield
     $D_v = \bar D(p, a, v)$.
                                                                                                                                                                                                                                                                                                Similarly, inserting the relations into \cref{eq:defrstd}, we obtain $R_v =  \bar R(p, a, v)$.
                                                                        
    On the other hand, consider the case $v \in [v_0, 1)$. Assuming $f(v_0) = 0$, we must have $f'(1) \ge 0$, which is equivalent to $1-p \le ap$. Furthermore, $v$ lies between two zeros of a convex function and is thus nonpositive, i.e., \cref{eq:cond2qv} is not satisfied.
    Thus, $q_v \ge 1$, which in turn yields $R_v = 0 = \bar R(p,a,v)$ and $D_v = 1-p = \bar D(p,a,v)$. 
    The case $f(v_0^{-1}) = 0$, i.e. $ap \le 1-p$, similarly yields $R_v = 0 = \bar R(p,a,v)$ and $D_v = ap = \bar D(p,a,v)$.
   
    Finally, observing that the case $v \in \{0,1\}$ follows by continuity concludes the proof of \cref{thm:iid_rd_final}.

\subsection{Proof of Theorem~\ref{th:subpop2}}
\label{app:proofsubpop2}

We first establish a general result for the joint information distortion rate function of several independent sources particularized to the setting of \cref{th:subpop2}. 
We will use the shorthand 
$D_{\textrm{I}}^n = D_{\textrm{I}}^{(X_n, \rho_n)}$ for the $n$-th individual and 
$D_{\textrm{I}}^{(i)} =  D_{\textrm{I}}^{(X^{(i)}, \rho^{(i)})}$, where $X^{(i)}, \rho^{(i)}$ are the random infection status and the cost function of an arbitrary individual in the $i$-th subpopulation.
\begin{lemma}
  \label{lem:dr_function}
  \begin{enumerate}
  The information distortion-rate function $D_{\textrm{I}}^{(\mathbf{X}, \rho)}(R)$ can be decomposed as 
    \begin{align}
      \label{eq:dr_decomp1}
      D_{\textrm{I}}^{(\mathbf{X}, \rho)}(R) &= \min_{\boldsymbol{\xi} \in [0,1]^N :\sum_n \xi_n = 1} \frac{1}{N} \sum_{n=1}^{N} D_{\textrm{I}}^n(\xi_n N R) \\
              &= \min_{\boldsymbol{\xi} \in [0,1]^{\mathcal I} : \sum_i \xi^{(i)}   N^{(i)} = 1} \frac{1}{N} \sum_{i \in \mathcal I} N^{(i)} D_{\textrm{I}}^{(i)}(\xi^{(i)} NR) . \label{eq:dr_decomp2}
    \end{align}
  \end{enumerate}
\end{lemma}
\begin{proof}
  In \cref{eq:DR_function}, the right-hand side can be rewritten as 
  \begin{align}
      \min_{p_{\mathbf{Y}|\mathbf{X}} : I(\mathbf{X}; \mathbf{Y}) \le NR} \mathbb E_{p_{\mathbf{X}} p_{\mathbf{Y}|\mathbf{X}}} \bigg[ \frac{1}{N} \rho(\mathbf{X}, \mathbf{Y}) \bigg]
      & = \min_{p_{\mathbf{Y}|\mathbf{X}} : I(\mathbf{X}; \mathbf{Y}) \le NR} \mathbb E_{p_{\mathbf{X}} p_{\mathbf{Y}|\mathbf{X}}} \bigg[ \frac{1}{N} \sum_{n=1}^{N} \rho_n(X_n, Y_n) \bigg]
      \notag \\
      & = \min_{p_{\mathbf{Y}|\mathbf{X}} : I(\mathbf{X}; \mathbf{Y}) \le NR} 
      \frac{1}{N} \sum_{n=1}^{N}
      \mathbb E_{p_{\mathbf{X}} p_{\mathbf{Y}|\mathbf{X}}} \big[ \rho_n(X_n, Y_n) \big]
      \notag \\
      & = \min_{p_{\mathbf{Y}|\mathbf{X}} : I(\mathbf{X}; \mathbf{Y}) \le NR} 
        \frac{1}{N} \sum_{n=1}^{N}
        \mathbb E_{p_{X_n} p_{Y_n|X_n}} \big[ \rho_n(X_n, Y_n) \big]\,.
        \label{eq:proofseprdfunc}
  \end{align}
  Furthermore, we can expand the mutual information as
  \begin{align}
    I(\mathbf{X}; \mathbf{Y})
    & = H(\mathbf{X}) - H(\mathbf{X}\,\vert\,\mathbf{Y})
      \notag \\
    & = \sum_{n=1}^N H(X_n) - H(X_n\,\vert\,\mathbf{Y}, X_1, \dots, X_{n-1})
      \notag \\
    & \geq \sum_{n=1}^N H(X_n) - H(X_n\,\vert\,Y_n)
      \notag \\
    & = \sum_{n=1}^N I(X_n;Y_n)
  \end{align}
  with equality if $p_{\mathbf{Y}|\mathbf{X}}(\mathbf{y}|\mathbf{x}) = \prod_n p_{Y_n|X_n}(y_n| x_n)$.
  We can thus restrict the minimization in \cref{eq:proofseprdfunc} to 
  $p_{\mathbf{Y}|\mathbf{X}}(\mathbf{y}|\mathbf{x}) = \prod_n p_{Y_n|X_n}(y_n| x_n)$ and obtain
  \begin{align}
    D_{\textrm{I}}^{(\mathbf X, \rho)}(R) &= \frac{1}{N} \min_{p_{\mathbf{Y}|\mathbf{X}} : \sum_n I(X_n; Y_n) \le NR} \sum_{n=1}^N \mathbb E_{p_{X_n} p_{Y_n|X_n}} \big[ \rho_n(X_n, Y_n) \big] \\
                                          &= \frac{1}{N} \min_{\xi \in [0,1]^N : \Vert v \Vert_1 = 1} \;\min_{p_{Y_n|X_n} : I(X_n; Y_n) \le \xi_n NR} \sum_{n=1}^N \mathbb E_{p_{X_n} p_{Y_n|X_n}} \big[ \rho_n(X_n, Y_n) \big] ,
  \end{align}
  which immediately yields \cref{eq:dr_decomp1}.
  
  Note that the only difference between \cref{eq:dr_decomp1,eq:dr_decomp2} is that all $\xi_n$ for $X_n$ belonging to the same subpopulation are chosen to be equal, i.e.,  $\xi_n = \xi^{(i)}$.
  To justify this choice, we have to show that this indeed minimizes the sum, i.e.,
  \begin{equation}
    N^{(i)} D_{\textrm{I}}^{(i)}(\xi^{(i)} NR) %
    \leq \sum_{n=1}^{N^{(i)}} D_{\textrm{I}}^{(i)}(\xi_n NR)
    \label{eq:convexsubpop}
  \end{equation}
  for $\sum_{n=1}^{N^{(i)}} \xi_n = N^{(i)} \xi^{(i)}$.
  The convexity of the information distortion-rate function \cite[Sec.~4.1]{gray90source} implies that
  \begin{equation}
    D_{\textrm{I}}^{(i)}\bigg(\sum_{n=1}^{N^{(i)}} \frac{1}{N^{(i)}} R_n\bigg)
    \leq \sum_{n=1}^{N^{(i)}} \frac{1}{N^{(i)}}  D_{\textrm{I}}^{(i)}(R_n)
  \end{equation}
  which is precisely \cref{eq:convexsubpop} for $R_n = \xi_n NR$.
\end{proof}

\Cref{th:subpop2} now follows by solving the minimization in \cref{eq:dr_decomp2}.
More specifically, we can write the optimization problem in \cref{eq:dr_decomp2} as an unconstrained optimization problem using a Lagrangian formalism \cite[Sec.~5.5.3]{Boyd2004Convex}
  \begin{align}
    \label{eq:2}
    \mathcal L(\boldsymbol{\xi}, \lambda, \boldsymbol{\mu}) = \frac{1}{N} \sum_{i \in \mathcal I} N^{(i)} D_{\textrm{I}}^{(i)}(\xi^{(i)} NR) + \lambda \left( 1 - \sum_i \xi^{(i)} N^{(i)} \right) - \sum_i \mu^{(i)} \xi^{(i)}
  \end{align}
  and obtain the associated Karush-Kuhn-Tucker (KKT) conditions for $\lambda \in \mathbb R$ and $\boldsymbol \mu \in \mathbb R^{\mathcal I}$,
  \begin{align}
    R N^{(i)} \dot D_{\textrm{I}}^{(i)}(\xi^{(i)} NR) - \lambda N^{(i)} - \mu^{(i)} &= 0  \qquad \text{ for } i \in \mathcal I \label{eq:kkt1}\\*
    1 - \sum_i \xi^{(i)} N^{(i)} &= 0 \label{eq:sum_constraint}\\*
    \xi^{(i)} \ge 0, \;  \mu^{(i)} \ge 0 , \;  \xi^{(i)} \mu^{(i)} &= 0 \qquad \text{ for } i \in \mathcal I . \label{eq:ximu0}
  \end{align}
  By the convexity of $D_{\mathrm{I}}^{(i)}$ \cite[Sec.~4.1]{gray90source}, the minimization problem \cref{eq:dr_decomp2} is convex and the KKT conditions therefore necessary and sufficient.
  
  Recall that we actually want to solve the minimization problem for a given $v$ at the fixed $R$
  \begin{align}
    R(v) = \frac{1}{N}\sum_i  N^{(i)} \bar R(p^{(i)}, a^{(i)}, v^{b^{(i)}})\,.
  \end{align}
  Choosing $\xi^{(i)} = \frac{\bar R(p^{(i)}, a^{(i)}, v^{b^{(i)}})}{NR}$ now obviously satisfies  \cref{eq:sum_constraint}.
  To check  \cref{eq:kkt1} and \cref{eq:ximu0} for all $i\in \mathcal I$, we consider two cases.
  If $v^{b^{(i)}} < v^{(i)}_0$, where $v^{(i)}_0$ is defined as $v_0$ in \cref{thm:RD_function}, then \cref{rmk:DR_deriv} implies that  
  $\dot D_{\textrm{I}}^{(i)}\big(\bar R(p^{(i)}, a^{(i)}, v^{b^{(i)}})\big) = \frac{b^{(i)}}{\log v^{b^{(i)}}}$.
  Hence, setting $\mu^{(i)}=0$, \cref{eq:ximu0} is satisfied and \cref{eq:kkt1} reduces to 
  choosing $\lambda = \frac{R}{\log v}$,
  such that \cref{eq:kkt1} is satisfied for all $i$ with $v^{b^{(i)}} < v^{(i)}_0$.
  
  If $v^{b^{(i)}}\geq v^{(i)}_0$, we obtain from the definition of $\bar R(p^{(i)}, a^{(i)}, v^{b^{(i)}})$ that $\xi^{(i)} = 0$.
  According to \cref{rmk:DR_deriv}, the derivative $\dot D_{\textrm{I}}^{(i)}(0) = \frac{b^{(i)}}{\log v_0^{(i)}}$ and  \cref{eq:kkt1} becomes $\mu^{(i)} = R N^{(i)} \frac{b^{(i)}}{\log v_0^{(i)}} - R N^{(i)} \frac{1}{\log v}$ which is nonnegative due to $v^{b^{(i)}}\geq v^{(i)}_0$ and thus  \cref{eq:ximu0} is satisfied.
  
  Thus, the choices $\xi^{(i)} = \frac{\bar R(p^{(i)}, a^{(i)}, v^{b^{(i)}})}{NR}$, $\lambda = \frac{R}{\log v}$, and $\mu^{(i)} = R N^{(i)} \frac{b^{(i)}}{\log v_0^{(i)}} - R N^{(i)} \frac{1}{\log v}$ for $v^{b^{(i)}}\geq v^{(i)}_0$ and $\mu^{(i)}=0$ otherwise, satisfy the KKT conditions. 
  The equations \cref{eq:subpopDR} then follow by inserting $\xi^{(i)} NR = \bar R(p^{(i)}, a^{(i)}, v^{b^{(i)}})$ into $D_{\textrm{I}}^{(i)}(\xi^{(i)} NR)$ and noting that, by \cref{thm:iid_rd_final}, we have $D_{\textrm{I}}^{(i)}(\bar R(p^{(i)}, a^{(i)}, v^{b^{(i)}})) = b^{(i)} \bar D(p^{(i)}, a^{(i)}, v^{b^{(i)}})$.

\section{Evaluation of $k$ stage group testing}
\label{app:proofksg}

To calculate the expected rate, we denote by $T_{\ell}$ the \ac{TpI} used in the $\ell$-th stage of the testing procedure.
Then 
\begin{equation}
    R_{\textrm{$k$SG($u_1, \dots, u_k$)}}
    = \mathbb{E}\bigg[\sum_{\ell =1}^{k} T_{\ell} \bigg]
    = \sum_{\ell =1}^{k} \mathbb{E}[ T_{\ell} ]
    .
    \end{equation}
Since in the first stage all individuals are tested in groups of size $u_1$, the expectation is simply
$
    \mathbb{E}[ T_{1} ]
    = \frac{1}{u_1}.
$
In the $(\ell+1)$-th stage for $\ell\geq 1$ only those individuals that belong to groups of size $u_{\ell}$ including at least one positive individual are tested. 
We can thus condition the expectation on the event that the group  of size $u_{\ell}$ the individual belongs to is negative, which has probability $(1-p)^{u_\ell}$, or positive, which has probability $1-(1-p)^{u_\ell}$.
Hence, the expectation expands to
\begin{equation}
    \mathbb{E}[ T_{\ell+1} ]
    = \big[ 1-(1-p)^{u_\ell} \big] \frac{1}{u_{\ell+1}} + (1-p)^{u_\ell} \cdot 0
\end{equation}
concluding the proof of \cref{eq:rateksg}.

To prove \cref{eq:distksg},
first note that $k$SG($u_1, \dots, u_k$) never declares an infected individual as healthy. 
Indeed, the only cases of wrong assignment can happen in groups that are positive at the $k$-th stage. 
Here, all individuals in the group are declared infected, although there may be several healthy individuals in the group. 
Thus, the expected cost of a single individual is the probability of itself being negative times the probability that at least one of $u_k-1$ other individuals is infected times the false positive cost $b$, i.e., 
\begin{align}
    D_{\textrm{$k$SG($u_1, \dots, u_k$)}} 
    & = b (1-p) (1-(1-p)^{u_k-1}) 
    \notag \\
    & = b (1 - p - (1-p)^{u_k})\,.
\end{align}

\end{appendix}

\section*{Acknowledgments}
We thank the anonymous reviewers for their valuable feedback and comments.

This work was supported in part by the  Vienna Science and Technology Fund (WWTF) under grant  MA16-053 and the Austrian Science Fund under grant Y 1199.
{\it Conflict of Interest}: None declared.

\bibliographystyle{elsarticle-num-names}
\bibliography{covid}

\end{document}